\documentclass[copyright,creativecommons]{eptcs}
\providecommand{\event}{QPL 2013}
\author{
Daniel Marsden
\institute{
Department of Computer Science\\
University of Oxford\\
}
\email{daniel.marsden@cs.ox.ac.uk}
}

\title{Fibred Coalgebraic Logic and Quantum Protocols}
\newcommand{\titlerunning}{Fibred Coalgebraic Logic and Quantum Protocols}
\newcommand{\authorrunning}{Daniel Marsden}

\usepackage[theorems]{macros}
\usepackage{names}
\usepackage{bussproofs}
\usepackage[numbers]{natbib}

\newcommand{\minorproof}[1]{}

\newcommand{\Alice}{\ensuremath{\operatorname{Alice}}\xspace}
\newcommand{\Bob}{\ensuremath{\operatorname{Bob}}\xspace}
\newcommand{\Both}{\ensuremath{\operatorname{Both}}\xspace}
\newcommand{\Bell}{\ensuremath{\operatorname{Bell}}\xspace}
\newcommand{\Channel}{\ensuremath{\operatorname{Channel}}\xspace}

\newcommand{\qdfuncn}[1]{\ensuremath{Q^d_{#1}}}

\newcommand{\distfunc}{\ensuremath{D}\xspace}
\newcommand{\san}[1]{\ensuremath{\mathcal{A}_{#1}}}

\newcommand{\predlift}[1]{\ensuremath{\llbracket #1 \rrbracket}}

\newcommand{\thenlift}[2]{\ensuremath{#2 \lhd #1}}

\newcommand{\deqmod}[1]{\ensuremath{\operatorname{Eq}_{#1}}}

\newcommand{\detcertmod}[1]{\ensuremath{\operatorname{Next}_{#1}}}
\newcommand{\dreqmod}[2]{\ensuremath{\operatorname{Eq}_{#1,#2}}}

\newcommand{\deqlift}[1]{\predlift{\deqmod{#1}}}

\newcommand{\detcertlift}[1]{\predlift{\detcertmod{#1}}}
\newcommand{\dreqlift}[2]{\predlift{\dreqmod{#1}{#2}}}

\newcommand{\qdeqmod}[3]{\ensuremath{\operatorname{Eq}_{#1,#2,#3}}}

\newcommand{\qdcertmod}[2]{\ensuremath{\operatorname{C}_{#1,#2}}}

\newcommand{\qdeqlift}[3]{\predlift{\qdeqmod{#1}{#2}{#3}}}

\newcommand{\eval}[1]{\ensuremath{ev^{#1}}}

\newcommand{\bellstate}[1]{\ensuremath{\psi_{#1}}}
\newcommand{\bellproj}[1]{\ensuremath{\hat{P}_{\bellstate{#1}}}\xspace}
\newcommand{\bellsadj}{\ensuremath{\hat{A}_{\Bell}}\xspace}

\newcommand{\lang}[2]{\ensuremath{\mathcal{L}^{#1}_{#2}}}
\newcommand{\typedlang}[3]{\ensuremath{\lang{#1}{#2}(#3)}}

\newcommand{\fibsigcat}[1]{\ensuremath{\mathcal{C}^{#1}}}
\newcommand{\fibsigsig}[2]{\ensuremath{\Lambda^{#1}_{#2}}}

\newcommand{\modtype}[2]{\ensuremath{#1^{#2} \rightarrow #1}}
\newcommand{\unarymodtype}[1]{\ensuremath{#1 \rightarrow #1}}
\newcommand{\bits}[1]{\ensuremath{\operatorname{[#1]}}}
\newcommand{\adapt}[1]{\ensuremath{\operatorname{#1}}}

\begin{document}
\maketitle

\begin{abstract}
Motivated by applications in modelling quantum systems using coalgebraic
techniques, we introduce a fibred coalgebraic logic.
Our approach extends the conventional predicate lifting semantics
with additional modalities relating conditions on different fibres. 
As this fibred setting will typically involve multiple signature functors, the
logic incorporates a calculus of modalities enabling the construction of
new modalities using various composition operations.
We extend the semantics of coalgebraic logic to this setting, and prove that
this extension respects behavioural equivalence.

We show how properties of the semantics of modalities are preserved under composition operations,
and then apply the calculational aspect of our logic to produce an expressive set of modalities for reasoning about
quantum systems, building these modalities up from simpler components.
We then demonstrate how these modalities can describe some standard quantum protocols.
The novel features of our logic are shown to allow for a uniform description of unitary evolution,
and support local reasoning such as ``Alice's qubit satisfies condition $\varphi$'' as is
common when discussing quantum protocols.
\end{abstract}

\section{Introduction}
In \citep{Abramsky2010} a coalgebraic model of quantum systems was constructed using
a novel fibrational structure to introduce ``enough contravariance'' to represent the
important physical symmetries of a quantum system. The paper then raised the question of what
a suitable ``fibred coalgebraic logic'' would look like, and that is the question we address
in this paper.

In the first half of the paper
we propose an extension of coalgebraic logic based upon predicate liftings \citep{Pattinson2003, Schroder2008}
(see also the excellent introduction \citep{Pattinson2008})
which provides a convenient setting in which to produce practical modal logics in a lightweight manner.
New types of modalities are introduced that allow explicit reasoning between different fibres,
and composition operations are provided to build modalities from simpler components.
In the second half of the paper we exploit the calculational aspects of our logic to construct
modalities suitable for reasoning about quantum protocols. The new features of our logic provide mechanisms
for describing important features such as unitary evolution, restriction to subsystems and local measurements.
Finally, we illustrate these features by applying them to two standard quantum protocols.

Fibred constructions involving coalgebras are also considered in \citep{KurzPattinson2000}  and \citep{KurzPattinson2000b},
in order to capture parameterization of signature functors. The question of fibred coalgebraic
logic using predicate liftings is explored in the later paper, but primarily from the perspective of the 
relationship to the logical structure of institutions \citep{GoguenBurstall1992} and this question is 
further pursued in \citep{Pattinson2002}. In contrast to the work in this paper, 
the logic discussed in these papers is exactly a conventional coalgebraic logic in each fibre, 
and the relationship between the fibres does not appear directly in the syntax of the logic.
In \citep{Marsden2013} a pseudo coalgebraic setting was introduced for modelling quantum systems, 
in order to develop the representation result of \citep{Abramsky2010} in a simpler and
more easily motivated setting. A coalgebraic logic was discussed in this setting, supporting a single
signature functor and modalities induced by its natural isomorphisms.

\section{Fibred Coalgebraic Logic}
\label{sec:fiblog}
Each fibre of our modal logic will correspond to a different signature functor. 
A fibred signature will describe a basic set of modalities that are available on each fibre. 

\begin{definition}[Modal Signature]
A \define{modal signature} $\Lambda$ is a set of modality symbols, each with
an associated cardinal referred to as the arity of the modality.
\end{definition}

\begin{definition}[Fibred Modal Signature]
A \define{fibred modal signature} $\Phi$ is a small monoidal category $\fibsigcat{\Phi}$ and for each
object $A$ in $\fibsigcat{\Phi}$ an associated modal signature $\fibsigsig{\Phi}{A}$. For
each pair of objects $A, B$ with $A \neq B$ we require that $\fibsigsig{\Phi}{A} \cap \fibsigsig{\Phi}{B} = \emptyset$.
\end{definition}

Given the basic set of modalities provided by the fibred modal signature, additional modalities
can be constructed via various composition operations.

\begin{definition}[Modality Expressions]
Let $\Phi$ be a fibred modal signature. We inductively define a typed
language of \define{modality expressions}, with conjunctions bounded
by a maximum cardinality $\kappa$.

We have one introduction rule:
\begin{center}
\begin{prooftree}
\AxiomC{$\Box_{\lambda} \in \fibsigsig{\Phi}{A}$ with arity $\alpha$}
\UnaryInfC{$\Box_{\lambda} : A^\alpha \rightarrow A$}
\end{prooftree}
\end{center}
We can apply logical operations to modality expressions:
\begin{center}
\begin{prooftree}
\AxiomC{$\bigcirc : \modtype{A}{\alpha}$}
\UnaryInfC{$\neg \bigcirc : \modtype{A}{\alpha}$}
\end{prooftree}
\end{center}
\begin{center}
\begin{prooftree}
\AxiomC{$0 < \card{I} < \kappa$ and $\bigcirc_i : \modtype{A}{\alpha}$ for each $i \in I$}
\UnaryInfC{$\bigwedge_{i \in I} \bigcirc_i : \modtype{A}{\alpha}$}
\end{prooftree}
\end{center}
We have 2 rules for constructing new modality expressions by composition:
\begin{center}
\begin{tabular}{p{5cm} p{5cm}}
\begin{prooftree}
\AxiomC{$\bigcirc_1 : \modtype{A}{\alpha}$}
\AxiomC{$\bigcirc_2 : B \rightarrow B$}
\BinaryInfC{$\thenlift{\bigcirc_1}{\bigcirc_2} : \modtype{(B \otimes A)}{\alpha}$}
\end{prooftree}
&
\begin{prooftree}
\AxiomC{$\bigcirc : \modtype{B}{\alpha}$}
\AxiomC{$f \in \fibsigcat{\Phi}(A, B)$}
\BinaryInfC{$\bigcirc^f : \modtype{A}{\alpha}$}
\end{prooftree}
\end{tabular}
\end{center}
\end{definition}

The formulae applicable on each fibre are described by mutual induction, allowing
the application of appropriate modality expressions as modalities:

\begin{definition}[Syntax and Typing]
For a fibred modal signature $\Phi$
we now define a language of typed formulae.
We write $\varphi : A$ for formula $\varphi$ is
of type $A$, in which case we will refer to $\varphi$ as an \define{$A$-formula}.

Our language is defined inductively by the following rules, starting with the typing rules for 
standard logical connectives for $A$ an object in  $\fibsigcat{\Phi}$:
\begin{center}
\begin{tabular}{p{1.5cm} p{1.5cm} p{5.5cm}}
\begin{prooftree}
\AxiomC{}
\UnaryInfC{$\top^A : A$}
\end{prooftree}
&
\begin{prooftree}
\AxiomC{$\varphi : A$}
\UnaryInfC{$\neg \varphi : A$}
\end{prooftree}
&
\begin{prooftree}
\AxiomC{$\varphi_i : A \mbox{ for each } i \in I \mbox { and } 0 < \card{I} < \kappa$}
\UnaryInfC{$\bigwedge (\varphi_i)_{i \in I} : A$}
\end{prooftree}
\end{tabular}
\end{center}
We have two application rules for the different types of modalities:
\begin{center}
\begin{tabular}{p{3.7cm} p{4.5cm}}
\begin{prooftree}
\AxiomC{$\varphi : B$}
\AxiomC{$f \in \fibsigcat{\Phi}(A, B)$}
\BinaryInfC{$f \varphi : A$}
\end{prooftree}
&
\begin{prooftree}
\AxiomC{$\varphi_i : A \mbox{ for each } i \in \alpha$}
\AxiomC{$\bigcirc : \modtype{A}{\alpha}$}
\BinaryInfC{$\bigcirc (\varphi_i)_{i \in \alpha} : A$}
\end{prooftree}
\end{tabular}
\end{center}
Modalities of the form $f$ for $f$ a \fibsigcat{\Phi} morphism will be referred to as \define{adaptation modalities}.
These modalities permit lifting of subformulae from different fibres in a suitable manner.

We will write $\lang{\Phi}{\kappa}$ for the formulae with conjunctions of cardinality at most $\kappa$ and
$\typedlang{\Phi}{\kappa}{A}$ for the $A$-formulae with conjunctions of cardinality at most $\kappa$.
\end{definition}

\begin{remark}
The category $[\refcset, \refcset]$ of endofunctors on \refcset and natural transformations
between them can be given the structure of a strict monoidal category, with the tensor given
by functor composition.
\end{remark}

\begin{definition}
We will write $\contrapowerset : \refcset^{op} \rightarrow \refcset$ for the contravariant powerset functor.
Define natural transformation $\neg : 2 \Rightarrow 2$ on components as:
\begin{equation}
\neg_X(U) := X \setminus U
\end{equation}
For each set $I$ define natural transformation $\bigwedge : 2^I \Rightarrow 2$ on components as:
\begin{equation}
\bigwedge_X((X_i)_{i \in I}) := \bigcap_{i \in I} X_i
\end{equation}
\end{definition}

The semantics for our logic are described by providing a structure identifying types with signature
functors, and the morphisms between types as suitable natural transformations. The tensor product
then corresponds to the composition of signature functors.

\begin{definition}[Structure]
For a given fibred modal signature $\Phi$, a \define{$\Phi$-structure} $S$ is a strict monoidal functor
$\llbracket - \rrbracket^S : \fibsigcat{\Phi} \rightarrow [\refcset, \refcset]$, and
for each object $A$ in $\fibsigcat{\Phi}$ and modality $\Box_\lambda$ in $\fibsigsig{\Phi}{A}$ of
arity $\alpha$ an associated natural transformation $\llbracket \Box_\lambda \rrbracket^S : 2^\alpha \Rightarrow 2 \circ \llbracket A \rrbracket^S$, referred to as a \define{predicate lifting} of arity $\alpha$.
\end{definition}

\begin{remark}
For a given fibred monoidal signature $\Phi$, the category $\fibsigcat{\Phi}$ will often be
a monoidal subcategory of $[\refcset, \refcset]$, with the functor 
$\llbracket - \rrbracket : \fibsigcat{\Phi} \rightarrow [\refcset, \refcset]$ given by the inclusion.
In later sections we will often identify the two when this is assumed to be the case.
\end{remark}

\begin{definition}[Modality Expression Semantics]
\label{def:modsem}
The semantics of modality expressions are given by suitable predicate liftings.
Let $\Phi$ be a fibred modal signature and $S$ a $\Phi$-structure.
Assume that $\alpha$ is a cardinal, $A,B$ are objects of \fibsigcat{\Phi}, $\Box_\lambda \in \fibsigsig{\Phi}{A}$,
$f : B \rightarrow A$ is a \fibsigcat{\Phi} morphism,
$\bigcirc : \modtype{A}{\alpha}$, for each $i \in I$ $\bigcirc_i : \modtype{A}{\alpha}$ and $\bigcirc' : \unarymodtype{B}$.
The semantics for modality expressions are given inductively as follows:
\begin{align}
\llbracket \Box_\lambda \rrbracket &:= \llbracket \Box_\lambda \rrbracket^S \\
\llbracket \neg \bigcirc \rrbracket &:= (\neg * \llbracket A \rrbracket^S) \circ \llbracket \bigcirc \rrbracket\\
\llbracket \bigwedge_{i \in I} \bigcirc_i \rrbracket &:= (\bigwedge * \llbracket A \rrbracket^S) \circ \langle \llbracket \bigcirc_i \rrbracket \mid i \in I \rangle \\
\llbracket \bigcirc^f \rrbracket &:= (2 * \llbracket f \rrbracket^S) \circ \llbracket \bigcirc \rrbracket\\
\llbracket \thenlift{\bigcirc}{\bigcirc'} \rrbracket &:= (\llbracket \bigcirc' \rrbracket * \llbracket A \rrbracket^S) \circ \llbracket \bigcirc \rrbracket
\end{align}
Above $\circ$ and $*$ denote vertical and horizontal composition of natural transformations respectively.
\end{definition}

\begin{definition}[Semantics of $A$-formulae]
Let $\Phi$ be a fibred modal signature and $S$ a $\Phi$-structure.
Assume $\alpha$ is a cardinal, $A$ is an object of \fibsigcat{\Phi}, $\bigcirc : \modtype{A}{\alpha}$ is a modality expression,
and $f : A \rightarrow B$ a \fibsigcat{\Phi} morphism.
The semantics for a formula $\varphi : A$, is given inductively for $\llbracket A \rrbracket$-coalgebra $(X, \gamma)$ as follows:
\begin{align}
\llbracket \top^A \rrbracket_{X, \gamma} &:= X\\
\llbracket \neg \varphi \rrbracket_{X, \gamma} &:= X \setminus \llbracket \varphi \rrbracket_{X, \gamma}\\
\llbracket \bigwedge (\varphi_i)_{i \in I} \rrbracket_{X, \gamma} &:= \bigcap_{i \in I} \llbracket \varphi_i \rrbracket_{X, \gamma}\\
\llbracket \bigcirc (\varphi_i)_{i \in \alpha} \rrbracket_{X, \gamma} &:= \gamma^{-1} \circ \llbracket \bigcirc \rrbracket_X ((\llbracket \varphi_i \rrbracket_{X, \gamma})_{i \in \alpha})\\
\llbracket f \varphi \rrbracket_{X, \gamma} &:= \llbracket \varphi \rrbracket_{X, \llbracket f \rrbracket^S_X \circ \gamma}
\end{align}
\end{definition}

\begin{remark}
The obvious relationships hold between logical operations on modality expressions and logical operations
on formulae. Also the logical operations commute appropriately with adaption modalities. We will
not need these properties for our examples, so the details are omitted.
\end{remark}

We now define a translation that will produce an equivalent formula with adaptation modalities removed. This
will allow use to reduce questions in the extended syntax to questions in the well understood setting of
coalgebraic logic with predicate liftings.

\begin{definition}[Translation]
For a given fibred modal signature $\Phi$,
for $f: A \rightarrow B$ in $\fibsigcat{\Phi}$, define the syntax translation $\tau_{f}$ as follows:
\begin{align}
\tau_{f}(\top^B : B) &:= \top^A : A\\
\tau_{f}(\neg \varphi : B) &:= \neg \tau_{f}(\varphi) : A\\
\tau_{f}(\bigwedge (\varphi_i)_{i \in I} : B) &:= \bigwedge (\tau_{f}(\varphi_i))_{i \in I} : A\\
\tau_{f}(\bigcirc (\varphi_i)_{i \in \alpha} : B) &:= \bigcirc^f (\tau_{f}(\varphi_i))_{i \in \alpha} : A\\
\tau_{f}(f' \varphi : B) &:= \tau_{f' \circ f}(\varphi) : A
\end{align}
\end{definition}

\begin{proposition}
\label{prop:semtrans}
For a given fibred modal signature $\Phi$ and $\Phi$-structure, for $f: A \rightarrow B$ in $\fibsigcat{\Phi}$:
\begin{equation}
\llbracket \varphi \rrbracket_{X, \llbracket f \rrbracket_X \circ \gamma} = \llbracket \tau_{f}(\varphi) \rrbracket_{X, \gamma}
\end{equation}
\end{proposition}

\begin{theorem}
The semantics of fibred coalgebraic logic respects behavioural equivalence.
\end{theorem}
\begin{proof}
By setting $f$ to the identity in proposition \ref{prop:semtrans} we get:
\begin{equation}
\llbracket \varphi \rrbracket_{X, \gamma} = \llbracket \tau_1(\varphi) \rrbracket_{X, \gamma}
\end{equation}
So the semantics of fibred coalgebraic logic is equivalent to the semantics of suitable formulae
in standard coalgebraic logic with predicate liftings, and this respects behavioural equivalence.
\end{proof}

\begin{example}[Simple combination of modality expressions]
\label{ex:kripke}
For a unary functor $F: \refcset \rightarrow \refcset$, and arbitrary set $A$, for each $a \in A$ we have
an obvious evaluation natural transformation $\eval{a} : F(-)^A \Rightarrow F(-)$.

Now for signature functor $\mathcal{P}$ (the powerset functor), giving Kripke frames as coalgebras, the
semantics of the usual $\Box$ modality is given by the following predicate lifting:
\begin{equation}
\llbracket \Box \rrbracket_X(U) := \mathcal{P}(U)
\end{equation}
If we consider the signature functor $\mathcal{P}(-)^A$ for (unbounded) labelled transition systems, the usual
$\Box_a$ modality can be constructed as the modality expression $\Box^{\eval{a}}$
\end{example}

\subsection{Semantics of Modality Expressions}
\label{sec:semmod}
In this section we consider some properties of predicate liftings such as monotonicity, continuity and being a separating
set, and how this is preserved under some of the composition operations described in section \ref{sec:fiblog}. We restrict
our attention to unary predicate liftings to simplify the presentation.

\begin{lemma}
\label{lem:compose}
Let $\Phi$ be a fibred modal signature and $S$ be a $\Phi$-structure. Let $\bigcirc : A \rightarrow A$ and $\bigcirc' : B \rightarrow B$
be modality expressions. Then if $\llbracket \bigcirc \rrbracket$ and $\llbracket \bigcirc' \rrbracket$ are monotone (continuous) then
$\llbracket \thenlift{\bigcirc}{\bigcirc'} \rrbracket$ is monotone (continuous).
\end{lemma}

\begin{lemma}
\label{lem:extend}
Let $\Phi$ be a fibred modal signature and $S$ a $\Phi$-structure. Let $\bigcirc : A \rightarrow A$ be a modality expression
and $f: B \rightarrow A$ a \fibsigcat{\Phi} morphism. Then if $\llbracket \bigcirc \rrbracket$ is monotone (continuous) then 
$\llbracket \bigcirc^f \rrbracket$ is monotone (continuous).
\end{lemma}

We now consider how expressive sets of predicate liftings are preserved under various operations. 
Results of this type are known and described in \citep{Pattinson2001}. We provide some results here
for completeness and in a form suitable for application in later examples.

\minorproof{
\begin{lemma}
\label{lem:prodtech}
Let $(F_i)_{i \in I} : \refcset \rightarrow \refcset$ be a family of endofunctors, $(\lambda^i : 2 \Rightarrow 2 \circ F_i)_{i \in I}
$ a family of predicate liftings, and 
and $\pi^i : \prod_{j \in I} F_j \Rightarrow F_i$ the projection natural transformations. We have:
\begin{equation}
p \in ((2 * \pi^i) \circ \lambda^i)_X(U) \iff \pi^i_X(p) \in \lambda^i_X(U)
\end{equation}
\end{lemma}
\minorproof{
\begin{proof}
By expanding definitions:
\begin{equation}
((2 * \pi^i) \circ \lambda^i)_X(U) = (\pi^i_X)^{-1}(\lambda^i_X(U))
\end{equation}
and the claim follows immediately.
\end{proof}
}
}

Expressivity can be lifted to products and exponentials from a fixed domain.

\begin{lemma}
\label{lem:prod}
Let $\Phi$ be a fibred modal signature and $S$ a $\Phi$-structure.
Let $(A_i)_{i \in I}$ be a family of objects in $\fibsigcat{\Phi}$. 
Assume $\llbracket A \rrbracket^S = \prod_{i \in I} \llbracket A_i \rrbracket^S$
and that there exist \fibsigcat{\Phi} morphisms $(\pi^i : A \rightarrow A_i)_{i \in I}$ such that $\llbracket \pi^i \rrbracket^S$
is the corresponding projection natural transformation.
For each $i \in I$ let $(\llbracket \bigcirc_{i,j} \rrbracket)_{j \in J_i}$
be a separating set of predicate liftings for $\llbracket A_i \rrbracket^S$. 
Then the predicate liftings $(\llbracket \bigcirc_{i,j}^{\pi_i} \rrbracket)_{i \in I, j \in J}$ are separating for 
$\llbracket A \rrbracket^S$.
\end{lemma}
\minorproof{
\begin{proof}
Assume $p, q \in \prod_{i \in I} \llbracket A_i \rrbracket^S(X)$ with $p \neq q$. Then there exists $\pi^i$ such that $\pi^i_X(p) \neq \pi^i_X(q)$.
Then by assumption without loss of generality there must be a $U \subseteq X$ and $\bigcirc^{i,j}$ such that:
\begin{equation}
\pi^i_X(p) \in \llbracket \bigcirc^{i,j} \rrbracket_X(U) \mbox{ and } \pi^i_X(q) \not \in \llbracket \bigcirc^{i,j} \rrbracket_X(U)
\end{equation}
the claim then follows from lemma \ref{lem:prodtech} and the semantics in definition  \ref{def:modsem}.
\end{proof}
}

\minorproof{
\begin{lemma}
\label{lem:evaltech}
Let $F : \refcset \rightarrow \refcset$ be an endofunctor, $\lambda : 2 \rightarrow 2 \circ F$ a
predicate lifting, $A$ a set and $a \in A$. Also let \eval{a} be as defined in example \ref{ex:kripke},
then:
\begin{equation}
g(a) \in \lambda_X(U) \iff g \in ((2 * \eval{a}) \circ \lambda)_X(U)
\end{equation}
\end{lemma}
\minorproof{
\begin{proof}
We have:
\begin{align}
g \in ((2 * \eval{a}) \circ \lambda)_X(U) &\iff  g \in (\eval{a}_X)^{-1} (\lambda_X(U))\\
 &\iff g(a) \in \lambda_X(U)
\end{align}
\end{proof}
}
}

\begin{lemma}
\label{lem:evalsep}
Let $\Phi$ be a fibred modal signature and $S$ a $\Phi$-structure.
Let $A,B$ be an objects in \fibsigcat{\Phi} with $\llbracket B \rrbracket = \llbracket A \rrbracket^A$.
Also let $(\eval{a} : B \rightarrow A)_{a \in A}$ be \fibsigcat{\Phi} morphisms such that $\llbracket \eval{a} \rrbracket$
is the corresponding evaluation natural transformation as defined in example \ref{ex:kripke}.
Let $(\llbracket \bigcirc_i \rrbracket)_{i \in I}$ be a
separating set of predicate liftings for $\llbracket A \rrbracket^S$. 
Then the predicate liftings
$(\llbracket \bigcirc_i^{\eval{a}} \rrbracket)_{i \in I, a \in A}$
are separating for  $\llbracket B \rrbracket^S$.
\end{lemma}
\minorproof{
\begin{proof}
Assume $g, g' \in \llbracket A \rrbracket(X)^A$ with $g \neq g'$. Then there must exist $a \in A$ such that
$g(a) \neq g'(a)$. By assumption, without loss of generality there exists $U \subseteq X$ and
$\llbracket \bigcirc_i \rrbracket$ such that $g(a) \in \llbracket \bigcirc_i \rrbracket_X(U)$ but $g'(a) \not \in \llbracket \bigcirc_i \rrbracket_X(U)$.The claim then follows from lemma \ref{lem:evaltech} and the semantics in definition \ref{def:modsem}.
\end{proof}
}

In general if we have separating sets of predicate liftings for two endofunctors, they do not combine (in any way) to
give a separating set for the composite functor. This is easily seen as, for example, the functor $\mathcal{P}_{\omega}$
has a separating set of liftings, but no separating set exists for $\mathcal{P}_{\omega} \circ \mathcal{P}_{\omega}$.
(See parts of (1) and (5) of example 23 in \citep{Schroder2008}). We examine a simple common case that we will require later, 
in which the behaviour is much better. The following notions will be useful:

\begin{definition}
Let $T : \refcset \rightarrow \refcset$ be an endofunctor.
Consider a set of predicate liftings $\{ \lambda^i \}$.
\begin{itemize}
 \item The liftings are said to \define{separate by singletons} if for an arbitrary set $X$, and $x,y \in T(X)$, 
  it is sufficient to consider the image of singleton sets under the $\lambda^i$ to separate $x$ and $y$.
 \item The liftings are said to be \define{mutually surjective on singletons} if for an arbitrary set $X$
 and each $t \in T X$ the singleton set $\{ t \}$ is in $\im{\lambda^i_X}$ for some $\lambda^i$.
\end{itemize}
\end{definition}

\begin{lemma}
For endofunctor $T : \refcset \rightarrow \refcset$,
any mutually surjective on singletons set of predicate liftings
is a separating set.
\end{lemma}
\minorproof{
\begin{proof}
For any set $X$ and $x,y \in T(X)$ with $x \neq y$ there exists $\lambda^i$ and $U \subseteq X$ such that:
\begin{equation}
\lambda^i_X(U) = \{x\}
\end{equation}
and then clearly:
\begin{equation}
x \in \lambda^i_X(U) \mbox{ and } y \not \in \lambda^i_X(U)
\end{equation}
\end{proof}
}

\begin{lemma}
\label{lem:sepcompsing}
Let $\Phi$ be a fibred modal signature and $S$ a $\Phi$-structure.
Let $A, B$ objects in \fibsigcat{\Phi}, $(\llbracket \bigcirc^B_i \rrbracket)_{i \in I}$
a set of predicate liftings on $\llbracket B \rrbracket$ that are mutually surjective on singletons, 
and $(\llbracket \bigcirc^A_j \rrbracket)_{j \in J}$ a separating
set of predicate liftings on $\llbracket A \rrbracket$ that separate by singletons. 
Then the liftings
$(\llbracket \thenlift{\bigcirc^B_i}{\bigcirc^A_j} \rrbracket) _{i \in I, j \in J }$
are separating for $\llbracket A \otimes B \rrbracket^S$.
\end{lemma}
\minorproof{
\begin{proof}
Let $x, y \in \llbracket A \rrbracket^S \circ \llbracket B \rrbracket^S (X)$ with $x \neq y$. Then without loss of generality, by assumption, there exists $\bigcirc^A_j$
and $\{u\} \subseteq S(X)$ such that:
\begin{equation}
x \in \llbracket \bigcirc^A_j \rrbracket_{S(X)}(\{u\}) \mbox{ and } y \not \in \llbracket \bigcirc^A_j \rrbracket_{S(X)}(\{u\})
\end{equation}
Now as the $(\llbracket \bigcirc^B_i \rrbracket)_{i \in I}$ are mutually surjective on singletons, 
there exist $V \subseteq X$ and $\bigcirc^B_i$ such that:
\begin{equation}
\llbracket \bigcirc^B_i \rrbracket_X(V) = \{ u \}
\end{equation}
and so:
\begin{equation}
x \in \llbracket \thenlift{\bigcirc^B_i}{\bigcirc^A_j} \rrbracket_X(V) \mbox{ and } y \not \in \llbracket \thenlift{\bigcirc^B_i}{\bigcirc^A_j} \rrbracket_X(V)
\end{equation}
\end{proof}
}

\section{Quantum Applications}
\label{sec:qapp}
We now consider a suitable signature functor for modelling quantum systems. In
\citep{Abramsky2010} a signature functor describing a ``question and answer system'' for projective measurements
was used.  We instead introduce a new functor based upon distributions of measurement 
outcomes for different physical quantities. When reasoning about quantum protocols it is
common to consider measurements in a suitable basis, rather than projective measurements, 
and this signature functor make the physical quantities and distribution
over measurement outcomes explicit.

\subsection{Constructing a Fibred Logic for Quantum Systems}
As an extended example, we construct an expressive set of modalities for reasoning about quantum systems
using simple components from well understood areas such as labelled transition systems and probabilistic logics.
An alternative modular approach to the construction of coalgebraic logics is presented in \citep{CirsteaPattinson2007}, 
based on a notion of syntax constructors.
Preservation of properties of modalities, such as expressivity, under operations including composition, products and
coproducts is analyzed in \citep{Pattinson2001}, and is probably closer in spirit to the approach of this section.
Many proofs are omitted throughout this section for space reasons, all conclusions are based upon
the composition based ideas in section \ref{sec:semmod} and standard results, mainly from \citep{Schroder2008}.

\begin{definition}
Let $\distfunc : \refcset \rightarrow \refcset$ denote the finite distribution functor, defined
on objects as follows:
\begin{equation}
D(X) := \{ f : X \rightarrow [0,1] \mid f \mbox{ has finite support and } \Sigma_{x \in X} f(x) = 1 \}\\
\end{equation}
and on morphisms:
\begin{equation}
D(f : X \rightarrow Y)(g \in D(X))(y \in Y) := \Sigma_{x \in X. f(x) = y} g(x)
\end{equation}
\end{definition}

\begin{lemma}
\label{lem:distomegaacc}
The finite distribution functor \distfunc is $\omega$-accessible.
\end{lemma}
\minorproof{
\begin{proof}
Let $X$ be an arbitrary set.
For any $f \in D(X)$ the function $f$ has finite support.
It is then easy to check that $f$ is in the image of $\distfunc(i)$, where $i$ is the
inclusion of the support of $f$ into $X$.
The result follows from proposition 5.2 in \citep{AdamekPorst2004}.
\end{proof}
}
Now we introduce our two basic building block modalities from which all others will be constructed.
\begin{lemma}
\label{lem:deqlift}
For the finite distribution functor \distfunc, for each $p \in [0,1]$ there is a unary predicate lifting 
$\deqlift{p} : 2 \Rightarrow 2 \circ \distfunc$
given by:
\begin{equation}
\deqlift{p}_X(U) := \{ d \mid \Sigma_{u \in U} d(u) = p \}
\end{equation}
These modalities separate by singletons.
\end{lemma}
\minorproof{
\begin{proof}
Checking naturality is straightforward. For separation by singletons,
let $X$ be some set and $f,g \in \distfunc{X}$ with $f \neq g$. Then there exists
$x \in X$ and $p \in [0,1]$ such that $f(x) = p \neq g(x)$. We then have that:
\begin{equation}
\deqlift{p}_X(\{ x \}) = \{ d \mid d(x) = p \}
\end{equation}
and so:
\begin{equation}
f \in \deqlift{p}(\{x\}) \mbox{ and } g \not \in \deqlift{p}(\{x\})
\end{equation}
\end{proof}
}

\begin{lemma}
\label{lem:detcertlift}
For a label set $\Sigma$, and $\sigma \in \Sigma$, define the unary predicate lifting $\detcertlift{\sigma} : 2 \Rightarrow 2 \circ (\Sigma \times (-))$ as follows:
\begin{equation}
\detcertlift{\sigma}(U) := \{ (\sigma, u) \mid u \in U \}
\end{equation}
These liftings are monotone and mutually surjective on singletons.
\end{lemma}
\minorproof{
\begin{proof}
That this is natural is trivial to check. Monotonicity is obvious from the definition.
For arbitrary $\{ (\sigma, x) \}$ we have:
\begin{equation}
\detcertlift{\sigma}_X(\{ x \}) = \{ (\sigma, x) \}
\end{equation}
giving the mutual surjectivity on singletons.
\end{proof}
}

Now we lift to distributions over eigenvalues.
\begin{lemma}
\label{lem:dreqlift}
For $p \in [0,1]$ and $r \in \mathbb{R}$ define predicate lifting $\dreqlift{p}{r} : 2 \Rightarrow \distfunc(\mathbb{R} \times (-))$
as the composite $\llbracket \thenlift{\detcertmod{r}}{\deqmod{p}} \rrbracket$.
This lifting is given explicitly by:
\begin{equation}
\dreqlift{p}{r}_X(U) := \{ d \mid \Sigma_{u \in U} d(r, u) = p \}
\end{equation}
These liftings are separating.
\end{lemma}
\minorproof{
\begin{proof}
The explicit form follows by expanding definitions.
For separation, by lemmas \ref{lem:detcertlift} and \ref{lem:deqlift} we can apply lemma \ref{lem:sepcompsing}.
\end{proof}
}

\begin{definition}
For finite dimensional Hilbert space $\mathcal{H}$ with dimension $n$,
let \san{n} denote the set of self adjoint operators.
Define the \define{distribution based quantum signature functor} \qdfuncn{n} as follows:
\begin{equation}
\qdfuncn{n} := \distfunc(\mathbb{R} \times (-) )^{\san{n}}
\end{equation}
There is an obvious \define{quantum coalgebra} for this signature, mapping pure states to distributions
over measurement outcomes and subsequent states.
\end{definition}

\begin{lemma}
\label{lem:dacc}
For a finite dimensional Hilbert space with dimension $n$,
the functor \qdfuncn{n} is accessible.
\end{lemma}
\minorproof{
\begin{proof}
We can describe our functor as the composite:
\begin{equation}
\qdfuncn{n} = (-)^{\san{n}} \circ \distfunc \circ (\mathbb{R} \times (-))
\end{equation}
Polynomial functors with exponents of cardinality less than $\lambda$ are $\lambda$-accessible.
As accessible functors are closed under composition, the claim then follows from lemma \ref{lem:distomegaacc}.
\end{proof}
}
Now we can lift to distributions for each self adjoint operator (physical quantity), giving a set of
liftings for our quantum signature functor $\qdfuncn{n}$:
\begin{lemma}
\label{lem:qdeqlift}
For finite dimensional Hilbert space $\mathcal{H}$ with dimension $n$,
for $p \in [0,1]$, $r \in \mathbb{R}$ and $\hat{A} \in \san{n}$ define unary predicate 
lifting $\qdeqlift{p}{r}{\hat{A}} : 2 \Rightarrow 2 \circ \qdfuncn{n}$ as follows:
\begin{equation}
\qdeqlift{p}{r}{\hat{A}} = \llbracket \dreqmod{p}{r}^{\eval{\hat{A}}} \rrbracket
\end{equation}
Where $\eval{\hat{A}}$ is as defined in example \ref{ex:kripke}.
These liftings are given explicitly by:
\begin{equation}
\qdeqlift{p}{r}{\hat{A}}_X(U) := \{ f \mid \Sigma_{u \in U } f(\hat{A})(r, u) = p \}
\end{equation}
and are separating.
\end{lemma}
\minorproof{
\begin{proof}
The explicit form is given by expanding definitions.
For separation, by lemma \ref{lem:dreqlift} the claim then follows from lemma \ref{lem:evalsep}.
\end{proof}
}

\begin{theorem}
For finite dimensional Hilbert space $\mathcal{H}$ with dimension $n$, any coalgebraic logic
with at least modalities with semantics given by the predicate liftings in lemma \ref{lem:qdeqlift}
is expressive if we allow conjunctions of sufficient cardinality.
\end{theorem}
\begin{proof}
By lemma \ref{lem:dacc} and proposition \ref{lem:qdeqlift} the claim follows immediately by applying theorem
14 of \citep{Schroder2008}.
\end{proof}

Although the unary predicate liftings based on equalities given in lemma \ref{lem:qdeqlift} are
very straightforward and separating, they are not monotone. It is easy to follow similar
steps to those above to construct a monotone set of modalities, based on lower bounds on the required
probabilities rather than equalities. 
This can be done for example by taking conjunctions of equality based modalities above the required threshold.
This gives an expressive logic using monotone modalities 
with semantics similar to those of
probabilistic modal logics \citep{LarsenSkou1991, HeifetzMongin2001}. For reasons of space, this
direction is not pursued further here as the equality based predicate liftings are sufficient for the quantum
protocols we will address.

In reality, although we have good expressivity results for the liftings above, they are not
particularly natural for the needs of describing quantum protocols. To aid reasoning about these
protocols, we would like our modalities to better match the actions that are performed during
their implementation. We now introduce some additional more ``practical'' modalities.

\begin{definition}
\label{def:projcert}
By noting that the natural transformations $\top : 1 \Rightarrow 2$ and $\neg : 2 \Rightarrow 2$ 
are predicate lifting for the identity functor, we can define $0$-ary modality:
\begin{equation}
\hat{P} := \thenlift{\top}{\qdeqmod{1}{1}{\hat{P}}}
\end{equation}
Intuitively, in the quantum model, this describes ``a projective measurement $\hat{P}$ is certain to have a positive outcome''.
We can also define unary modality:
\begin{equation}
\qdcertmod{r}{\hat{A}} := \thenlift{\neg}{\qdeqmod{0}{r}{\hat{A}}}
\end{equation}
with the reading ``it is certain that after getting measurement outcome $r$ when measuring physical quantity $\hat{A}$, $\varphi$ will hold''.
\end{definition}

\begin{definition}
Using similar tools to those above, we can combine our unary modalities to provide a possibilistic polyadic modality, 
describing how subsequent states relate to possible measurement outcomes:
\begin{align}
\hat{A}(r_1 \mapsto (-),...,r_n \mapsto (-)) 
\end{align}
Informally this has semantics ``after measuring $\hat{A}$, if outcome $r_i$ occurs then the $i^{th}$ postcondition will hold.''
\end{definition}

\subsection{Basic Quantum Operations}
We first consider how some of the features of our fibred coalgebraic logic can be applied
to describe notions commonly considered when analyzing quantum systems and protocols.

\begin{example}[Unitary Evolution]
\label{ex:unitary}
For an arbitrary Hilbert space $\mathcal{H}$ we consider the quantum signature functor. A unitary
$\hat{U}$ on $\mathcal{H}$ induces a function $\hat{P} \mapsto \hat{U} \hat{P} \hat{U}^{\dagger}$
giving a natural transformations $\qdfuncn{n} \Rightarrow \qdfuncn{n}$ by precomposition.
These give adaptation modalities in our fibred coalgebraic logic, which in the case of
the quantum coalgebra encode unitary (Heisenberg type) evolution of the system. In this approach
the unitary evolution is encoded \emph{uniformly} across each coalgebra without extending the signature functor. 
We will write $\operatorname{\hat{U}} \varphi$ for ``after applying unitary transformation $\hat{U}$, $\varphi$ holds''.
\end{example}

\begin{example}[Restriction to Subsystems]
\label{ex:sub}
We consider a 2 qubit quantum system, with corresponding Hilbert space $\mathcal{H}_2 \otimes \mathcal{H}_2$.
We then fix a basis and define a linear map $\ket{ij} \mapsto \ket{i}$,
and then define natural transformations $\Alice : \qdfuncn{4} \Rightarrow \qdfuncn{2}$ by precomposition with the inverse image
of this linear map. This natural transformation induces adaptation modalities in our logic such that we can read $\Alice \varphi$
as ``if we restrict our attention to Alice's qubit, $\varphi$ holds.'' Note that we have not needed to explicitly 
introduce mixed states to handle restriction to subsystems as this is encoded in the measurements selected by
the \Alice natural transformation.
\end{example}

\begin{example}[Local Measurements]
If we consider a single qubit system, 
the $\llbracket \hat{P} \rrbracket$ predicate lifting given in definition \ref{def:projcert} 
describes certainty of projective measurement $\hat{P}$. The natural transformation \Alice
defined in example \ref{ex:sub} then gives modality $\hat{P}^{\Alice}$ giving certainty
of measurement $\hat{P}$ locally on Alice's qubit in a 2 qubit composite system.
\end{example}

\subsection{Quantum Teleportation}
\begin{definition}
We will write $\bellproj{i}$ for the projection operator corresponding to the $i^{th}$ Bell state.
\end{definition}
We consider the standard example of the quantum teleportation protocol \citep{BennettBrassardCrepeauJozsaPeresWooters1993}.
This is a 3 qubit protocol that can be informally described as follows:
\begin{quote}
Initially Alice has a qubit in (arbitrary) state $\varphi$ and she also shares half of a two qubit pair in the Bell state (the channel) with Bob. After a Bell basis measurement on both of Alice's qubits, if Bob applies a suitable correcting unitary, dependent on the outcome of the measurement, he can be certain his qubit is now in state $\varphi$.
\end{quote}
We can formalize this in our logic as the following formula:
\begin{align}
\hat{P}_{\varphi}^{\Alice} \wedge \bellproj{1}^{\Channel} \Rightarrow \hat{A}_{\Bell}^{\Both}(r_1 &\mapsto \Bob \adapt{\hat{U}_1} \hat{P}_{\varphi},\\
 r_2 &\mapsto \Bob \adapt{\hat{U}_2} \hat{P}_{\varphi},\\
 r_3 &\mapsto \Bob \adapt{\hat{U}_3} \hat{P}_{\varphi},\\
 r_4 &\mapsto \Bob \adapt{\hat{U}_4} \hat{P}_{\varphi})
\end{align}
As our modality is built from a conjunction of smaller modalities, we can adopt a more ``post selection'' style perspective and 
decompose our teleportation protocol into various possible measurement outcomes. 
Here we consider formulae capturing each of the $i \in \{1..4\}$ measurement outcomes separately:
\begin{equation}
\hat{P}_{\varphi}^{\Alice} \wedge \bellproj{1}^{\Channel}  \Rightarrow \qdcertmod{r_i}{\hat{A}_{\Bell}}^{\Both}(\Bob \adapt{\hat{U}_i} \hat{P}_\varphi)
\end{equation}

\subsection{Entanglement Swapping}
\begin{definition}
To simplify notation for multi-qubit systems we will now write $\bits{i,j}$ for the restrict to bits $i$ and $j$, rather than define a proliferation of named subsystems such as \Alice, \Channel etc. as in the previous protocol and examples.
\end{definition}

We now consider the 4 qubit entanglement swapping protocol \citep{ZukowskiZeilingerHorneEkert1993},
informally this protocol can be summarized as:
\begin{quote}
Initially qubits 1 and 2, and qubits 3 and 4 are in the Bell state. After a measurement on qubits 2 and 3 in the Bell basis and applying suitable corrective unitaries, dependent on the measurement outcome, we can be certain to leave qubits 1 and 4 and qubits 2 and 3 in the Bell state.
\end{quote}
This can be encoded in our modal logic for the $i \in {1..4}$ measurement outcomes as formulae of the form:
\begin{equation}
\bellproj{1}^{\bits{1,2}} \wedge \bellproj{1}^{\bits{3,4}} \Rightarrow \qdcertmod{r_i}{\bellsadj}^{\bits{2,3}}(\bits{1,4} \adapt{\hat{U}_i} \bellproj{1} \wedge \bits{2,3} \adapt{\hat{U}_i'} \bellproj{1})\\
\end{equation}

\section{Conclusions and Future Work}
We have presented a fibred coalgebraic logic and shown that it respects behavioural equivalence.
A distribution based signature functor for modelling finite dimensional quantum systems was introduced and
the calculational aspects of our logic were exploited to construct suitable modalities for reasoning
about quantum protocols. 
It was shown that expressivity of the logic could be lifted via the composition
operations from modalities for simpler and well understood signature functors.
The fibred aspects of our logic were exploited to capture key components of quantum computation, including
a uniform description of unitary evolution, restriction to local subsystems and encoding of local measurements 
on  composite systems.

The current work primarily concerns semantics. Proof theoretic aspects, particularly their
suitability for analysis of quantum protocols, will be
pursued in later work.  The logic presented here seems to potentially be a special case
of a general construction that could be applied to a suitable class of institutions \citep{GoguenBurstall1992}, 
this should be investigated further. 
Connections to the existing automated tools in coalgebraic logic,
and their application to analyzing quantum protocols should also be pursued.

\subsubsection*{Acknowledgements}
I would like to thank Andreas D\"oring and Samson Abramsky for their feedback and suggestions.
I would also like to thank the anonymous referees for their valuable comments and detailed recommendations.

\bibliographystyle{eptcs}

\end{document}